\documentclass[12pt,oneside,reqno]{amsart}

\usepackage[final]{graphicx}
\usepackage{amsfonts}
\usepackage{pdfsync}
\usepackage{amsmath}
\usepackage{amssymb}
\usepackage{amsthm}
\usepackage{caption}
\usepackage{enumerate} 
\usepackage{dcolumn} 

\newtheorem{theorem}{Theorem}[section]
\newtheorem{lemma}[theorem]{Lemma}

\newtheorem{remark}[theorem]{Remark}
\newtheorem*{remarks}{Remarks}

\newcommand\TSC{\rule{0pt}{2.6ex}}       
\newcommand\BSC{\rule[-1.2ex]{0pt}{0pt}} 

\newcommand{\R}{{\mathord{\mathbb R}}}

\newcommand{\N}{{\mathord{\mathbb N}}}

\newcommand{\supp}{{\mathop{\rm supp\ }}}

\usepackage{color}

\newcommand\numberthis{\addtocounter{equation}{1}\tag{\theequation}}
\def\nn{\nonumber}

\newcommand{\di}{\,\mathrm{d}}

\newcommand{\abs}[1]{\left|#1\right|}
\newcommand{\norm}[1]{\lVert#1\rVert}
\newcommand{\bvec}[1]{\boldsymbol{#1}}

\newcommand{\inprodtwo}[2]{\left \langle #1 , #2\right \rangle}

\begin{document}
\title[GNS inequalities for convex domains in $\R^d$]{\textbf{Gagliardo-Nirenberg-Sobolev inequalities for convex domains in $\R^d$}}

\author[Benguria]{Rafael~D.~Benguria$^1$}

\author[Vallejos]{Cristobal~Vallejos$^2$}

\author[Van Den Bosch]{Hanne~Van~Den~Bosch$^3$}

\address{$^1$ Instituto  de F\'\i sica, Pontificia Universidad Cat\' olica de Chile,}
\email{{rbenguri@fis.puc.cl}}

\address{$^2$ Facultad de F\'\i sica, Pontificia Universidad Cat\' olica de Chile,}
\email{{civallejos@uc.cl}}

\address{$^3$ Centro de Modelamiento Matem\'atico, CMM, FCFM, Universidad de Chile}
\email{{hvdbosch@cmm.uchile.cl}}

\begin{abstract} 
A special type of Gagliardo--Nirenberg--Sobolev (GNS) inequalities in $\mathbb{R}^d$ has played a key role in several proofs of Lieb--Thirring inequalities. 
Recently, a need for GNS inequalities in convex domains of $\mathbb{R}^d$, in particular for cubes, has arised. The purpose of this manuscript is two--fold. First we 
prove a GNS inequality for convex domains, with explicit constants which depend on the geometry of the domain. Later, using 
the discrete version of Rumin's method, we prove GNS 
inequalities on cubes with improved constants. 
\end{abstract}

\maketitle

\section{Introduction} \label{intro}

In the past sixty years, following the original papers \cite{Gag959,Nir959}, there has been a huge literature on Gagliardo--Nirenberg inequalities. A particular 
case of this type of inequalities is the Gagliardo--Nirenberg--Sobolev inequality (GNS for short),  
\begin{equation}
\int_{\R^d} \abs{\nabla u}^2   \di x  \ge G(d) \left( \int_{\R^d} u^2   \di x \right)^{-2/d}   \int_{\R^d} u^{2(1+2/d)}    \di x ,
\label{GNSRDIntro}
\end{equation} 
which holds for $u \in H^1(\mathbb{R}^d)$, in any dimension $d \ge 1$. The inequality (\ref{GNSRDIntro}) is related to the embedding of the Sobolev space 
$H^1(\mathbb{R}^d)$ in $L_{2+4/d}(\mathbb{R}^d)$. The GNS inequality also arises in the context of Lieb--Thirring inequalities. In fact, if one is interested
in maximizing the absolute value,  to a power $\gamma$ say, of the ground state energy of the Schr\"odinger operator $H=-\Delta + V$, acting on $L^2(\mathbb{R}^d)$
keeping $\int_{\mathbb{R}^d} V_{-}(x)^{\gamma+d/2}   \di x$ fixed, one is immediately lead to consider an inequality like (\ref{GNSRDIntro}) (see \cite{LiTh976} for the original 
discussion, stemming from the Lieb--Thirring Conjecture; see also \cite{DoFeLoPa006}). Recently, an extension of (\ref{GNSRDIntro}) in graphs  has been considered 
\cite{AdSeTi017}. Of course, one of the challenging questions concerning Gagliardo-Nirenberg inequalities in general and the GNS equation (\ref{GNSRDIntro}) in particular 
is to determine sharp constants and to obtain good approximations to them. We will come back to this point at the end of this introduction. 

Several  authors have also  considered extensions of Gagliardo--Nirenberg inequalities to particular domains in $\mathbb{R}^d$ (see, e.g., 
\cite{AcDu003, Zug017}). In this category one can also consider the classical work of Payne and Weinberger \cite{PaWe960} and many results 
stemming from it.  

Recently  Phan--Th\`anh~Nam \cite{Nam018} used microlocal analysis to derive a sharp estimate for the expectation value of the kinetic energy of $N$ non relativistic fermions in terms 
of a functional on their  single particle density. The leading term of Nam's bound is the conjectured Lieb--Thirring bound for the kinetic energy with the semiclassical 
constant in front plus a correction term depending on the gradient of the single particle density. One of the tools used by Nam on his microlocal analysis is 
a GNS inequality of the form (\ref{GNSRDIntro}) but in a unit cube in $d$ dimensions instead of $\mathbb{R}^d$. 
Although in Nam's work no attention is paid to the 
value of the constant in his GNS inequality on cubes,
one needs to have good estimates on the corresponding constants to have an estimate on the gradient correction. 
It is precisely this need which motivates the present results. 
For a bounded domain $\Omega \subset \R^d$, we define 
\begin{equation}
\label{eq:var-prob}
G(\Omega, d) = \inf \frac{\int_{\Omega} \abs{\nabla u}^2 \di x   \left(\int_{\Omega} u^2  \di x  \right)^{2/d}}{\int_{\Omega} {\vert{u-u_\Omega \vert}}^{2+4/d}  \di x },
\end{equation}
where the infimum is taken over functions $u \in W^{1,1}(\Omega)$ and $u_\Omega$ is its average $\abs{\Omega}^{-1} \int_\Omega u$.
Since we are specially interested in the case $\Omega = [0,1]^d$, we will write $G([0,1]^d, d) \equiv G_Q(d)$.
Our main results are summarized in the following theorem.
\begin{theorem}
 \label{thm:meta}
 With the definitions in \eqref{GNSRDIntro} and \eqref{eq:var-prob}, the following holds.
 \begin{enumerate}[{\it i)}]
  \item \label{it:d=1} For $d=1$, \[G_Q(1) = G(1)/4 = \pi^2/16, \] and the infimum is not attained.
  \item \label{it:convex}For all $d \ge 3$ and all convex $\Omega$, we have
  \[
  G(\Omega, d) \ge \left( \frac{ d \, \abs{\Omega}}{ \operatorname{diam}(\Omega)^d \, C_{\rm HLS}(d,d-1,2) }\right)^2
  \]
  \item \label{it:cubes}For cubes and all $d \ge 2$, we have 
  \begin{equation}
   \label{eq:meta-cubes}
   \frac{G(d)}{4} \ge G_Q(d) \ge \frac{\pi^2 d^2}{(d+4)(d+2) N_d^{2/d}},
  \end{equation}
  and if the first inequality is strict, a minimizer exists.
 \end{enumerate}
Here $C_{\rm HLS}$ is the constant in the Hardy-Littlewood-Sobolev inequality (see Section~\ref{sec:Davies} for details), and $N_d$ is given in \eqref{eq:bound_Nd} in Section~\ref{sec:Rumin}.
\end{theorem}

The remainder of the manuscript contains the proof of the results stated above, divided in several separate statements. Section~\ref{sec:Davies} contains the proof of point~\ref{it:convex}), based on a result of Brian Davies \cite{Dav990} that holds for general convex domains. The next two sections deal with the special case of cubes. 
The lower bound of point~\ref{it:cubes}), is established in Section~\ref{sec:Rumin}, using a method originally due to Rumin \cite{Ru011}. This bound gives better numerical constants (see table \ref{tbl:cubes}) and has a simpler proof than the more general statement of point~\ref{it:convex}). In addition, it also holds for $d=2$. We have not been able to obtain bounds for other domains in $\R^2$.
The final Section~\ref{sec:GNS} contains the proofs of the one-dimensional case \ref{it:d=1}), of the upper bound in point \ref{it:cubes}) and of the results on (non)-existence. These last results rely on a rearrangement lemma for cubes that will be proven in the appendix.

\bigskip

Let us finish this introduction with a short overview of known bounds for the Gagliardo-Nirenberg-Sobolev constants in $\R^d$.
It is well known that the variational problem in \eqref{GNSRDIntro} has a unique minimizer (up to translation, scaling and multiplication), 
which is radially symmetric, decreasing, and  can be taken positive. The values of the constants $G(d)$ are only known for $d =1$ (see, e.g., \cite{Agu008,Lun017,Nas989} and 
references therein), where the value is $G(1)=\pi^2/4$ (for an alternative proof of this fact see \cite{BeLo004}).

The inequality \eqref{GNSRDIntro} is a particular case of a Gagliardo-Nirenberg type inequality that characterizes the embedding 
of $H^1(\mathbb{R}^d)$ in $L_{\rho+2}(\mathbb{R}^d)$, of the form
\begin{equation}
\|u\|_{\rho+2} \le k(\rho,d)  \|\nabla u\|_2^{\alpha} \|u\|_2^{1-\alpha}. 
\label{Besov1}
\end{equation}
Here, 
\begin{equation}
\alpha=\frac{d}{2} \frac{\rho}{\rho+2}.
\label{Besov2}
\end{equation}
The inequality (\ref{Besov1})  holds for any $\rho \in (0,\rho_0)$, where $\rho_0 =4/(d-2)$, if $d \ge 3$, and  $\rho=\infty$ if $d=1,2$
(see, e.g., \cite{BeIlNi975,Nas989}, and references therein). The inequality (\ref{GNSRDIntro}) corresponds to setting $\rho=4/d$ in 
(\ref{Besov1}). 

Except for particular cases, the optimal constant $k(\rho,d)$ for (\ref{Besov1}) is not known.  The best 
estimates to date for $k(\rho,d)$ are the ones obtained by Nasibov in \cite{Nas989}, namely, 
\begin{equation}
k(\rho,d)  \le k_N(\rho,d) \equiv \frac{1}{\chi} \left( \frac{|\mathbb{S}^{d-1}| \, B(\frac{d}{2},\frac{d(1-\alpha)}{2 \alpha})}{2}\right)^{\alpha/d} 
k_{BB}\left(\frac{\rho+2}{\rho+1} \right).
\label{Besov3}
\end{equation}
In (\ref{Besov3}), 
\begin{equation}
\chi = \sqrt{\alpha^{\alpha} \, (1-\alpha)^{1-\alpha}}, 
\label{Besov4}
\end{equation}
and $B(x,y)$ is the Euler Beta function, i.e., $B(x,y) = \Gamma(x) \Gamma(y)/\Gamma(x+y)$. 
Moreover, 
\begin{equation}
k_{BB} (p) =\left( \left(\frac{p}{2\pi}\right)^{1/p}\big/\left(\frac{p'}{2\pi}\right)^{1/p'} \right)^{d/2},
\label{Besov5}
\end{equation}
for $1<p<\infty$ and $1/p+1/p'=1$, is the optimal constant for the Hausdorff--Young inequality, as it was proven by Babenko \cite{Bab961} and Beckner
\cite{Bec975}. It follows from the previous discussion that (\ref{GNSRDIntro}) holds for any $u \in H^1(\R^d)$, for $d \ge 1$ and 
\begin{equation} 
G(d) \ge G_N(d) \equiv \bigl( k_N(4/d,d) \bigr)^{-2/\alpha}.
\label{Besov6}
\end{equation}

More recently, (\ref{GNSRDIntro}) has been also proven using a projection of the Fourier transform of $u$ into high and low energy components, a method inspired on 
Rumin's techniques \cite{Ru011}: see, e,g., \cite{Fra014, LuSo013, Sol011}. Using these techniques one can prove that
\begin{equation}
G(d)  \ge G'(d)  \equiv \frac{(2\pi)^2 d^{2+2/d} \abs{\mathbb{S}^{d-1}}^{-2/d}}{(d+2)(d+4)}.
\label{GNSconstant}
\end{equation}
A detailed proof and further comments and references can be found as Theorem 4.14 in the recent lecture notes  \cite{Lun017}.

As pointed out in \cite{Lun017} the optimal constant in (\ref{GNSRDIntro}) satisfies $G(1)=\pi^2/4$, $G(2) = S_{2,4}$ and $G(d)\ge S_d$ for all $d \ge 3$. 
Here $S_d=d(d-2)|\mathbb{S}^d|^{2/d}/4$ is the optimal constant in Sobolev's inequality
$$
\int_{\mathbb{R}^d} (\nabla u)^2   \di x \ge S_d {\|u\|^2_{2d/(d-2)}},
$$
which holds for all $u \in H^1(\mathbb{R}^d)$, and $d\ge 3$, while $S_{2,4}$ is the optimal constant of the inequality
$$
\int_{\mathbb{R}^2} (\nabla u)^2   \di x \ge S_{2,4} \|u\|_2^{-2} \|u\|_4^4. 
$$
The value of $S_{2,4}$ is not known, but there are 
well known lower bounds,
(see, e.g.,  \cite{LiLo001} Theorem 8.5).

Although the lower estimates on $G(d)$ obtained using Rumin's techniques are worse than the ones obtained by Nasibov, it is worth introducing them since their proof is simpler. Moreover, we will use the discrete version of Rumin's method to obtain lower bound on $G_Q(d)$ in Section~\ref{sec:Rumin}.
We have summarized the present situation concerning the numerical values of the estimates on $G(d)$ in Table $1$.

\begin{table}[h!]
\centering
\begin{tabular}{l | c| c | r}
$d$ & $G_N$(d) & $G'(d)$ & $G(d)$  \BSC\\

\hline

$1$ & $2.2705$ &    $0.6580$  &   $2.4674$ \TSC \\
$2$ & $5.3014$ &    $2.0944$  &    $ 5.850$ \hphantom{0} \\
$3$ & $8.6427$ &    $3.9067$  & $9.578$ \hphantom{0}\\
$4$ & $12.1605$ &   $5.9238$  &   $ 13.489$ \hphantom{0}\\ 
$5$ & $15.7941$ &   $8.0619$  &  $ 17.483$  \hphantom{\,} \BSC \\

\hline
\end{tabular}
\caption{The values of the second  column in this table are the bounds of Nasibov given by equations (\ref{Besov3}), (\ref{Besov6}) above, 
the values on the third column are the bounds obtained using Rumin's techniques (equation (\ref{GNSconstant})).
The fourth column contains the known exact value for $d=1$, i.e., $\pi^2/4$, whereas the values for $d\ge 2$ are obtained by us through numerical integration of 
the Euler equations associated with (\ref{GNSRDIntro}).} 
\label{tbl:GNS}
\end{table}

\section{Estimates for general convex domains in $\R^d$, $d \ge 3$}
\label{sec:Davies}
In this section we prove a Gagliardo-Nirenberg-Sobolev inequality on convex domains of $\R^d$.

\begin{theorem} \label{thm:GNSconvex}
If $\Omega$ is a bounded, convex set in $\mathbb{R}^d$, $d\ge 3$, and $u \in W^{1,1}(\Omega)$, we have 
$$
  \int_{\Omega} |\nabla u|^2  \di x   \left(\int_{\Omega} u^2 \di x\right)^{2/d}
  \ge C_{1}(\Omega,d) \int_{\Omega} |u-u_{\Omega}|^{2+4/d}  \di x.
$$
Here
\[
 C_1(\Omega, d) =(C_{\rm D}(\Omega, d) \,  C_{\rm HLS}(d,d-1,2) )^{-2},
\] 
where $C_{\rm D}(\Omega, d)=\operatorname{diam}(\Omega)^d/(d|\Omega|)$, is the geometric constant of Lemma \ref{average} below; 
$C_{\rm HLS}$ is the Hardy-Littlewood-Sobolev constant  (see, e.g., \cite{Lie983} and \cite{LiLo001} Theorem 4.3).
\end{theorem}

One of the ingredients in our proof of Theorem \ref{thm:GNSconvex} is the following result of Brian Davies whose proof we give here for completeness. This lemma has been used recently in similar inequalities in \cite{MiTaSeOi017}. Numerical values of the constants $C_1(\Omega,d)$ when $\Omega$ is a cube are given in Table \ref{tbl:cubes}. 

\begin{lemma} [E.~B.~Davies, \cite{Dav990}, Lemma 1.7.3] \label{average}
Let $\Omega$ be a convex, bounded domain in $\mathbb{R}^d$, with volume $|\Omega|$ and diameter $\operatorname{diam} ( \Omega)$. If $f\in W^{1,1}(\Omega)$, then, 
\begin{equation}
|f(x) - f_{\Omega}| \le \frac{\operatorname{diam} ( \Omega)^d}
{d|\Omega|} \left(h_{d/(d-1)}*|\nabla f| \right)(x),
\label{eq:Davies}
\end{equation}
almost everywhere. 
Here, 
$$
h_s(x) = |x|^{-d/s},
$$
and 
$$
(f*g)(x) = \int f(x-y) g(y)   \di y.
$$
\end{lemma}
\begin{remarks} 
\item{\rm i)} With a slight abuse of notation, here $\nabla f(x)$  is the extension by zero of $\nabla f$ in $\Omega$, to $\R^d$. 

\item{\rm ii)} The averaging method used in this proof is a standard tool in partial differential equations. This method  goes back to the proof of the \emph{Huygens principle} for the solution to the wave equation.
\end{remarks}

\begin{proof}
Since $C^{\infty} \cap W^{1,1}(\Omega)$ is dense in $W^{1,1}(\Omega)$ it is sufficient to consider smooth functions.
If $x, y \in \Omega$, using the Fundamental Theorem of Calculus and the convexity of the domain we have
\begin{equation}
f(x) - f(y) = - \int_0^{\rho}\frac{\partial}{\partial r} f(x+ r \omega)   \di r 
\label{eq:1}
\end{equation}
where $\omega=(y-x)/\abs{y-x}$, and $\rho=\abs{y-x}$.
For fixed $x \in \Omega$ we average over $y \in \Omega$ and obtain, 
\begin{equation}
|f(x) - f_{\Omega}| 
\le \frac{1}{|\Omega|} \int_{|\omega|=1} d\omega\int_0^{\rho_{max}(\omega)}\rho^{d-1}   \di \rho
\int_0^\rho \abs{\nabla f(x+r\omega)}   \di r.
\label{eq:2}
\end{equation}
In \eqref{eq:2}, $\rho_{max}(\omega)$ is the distance of $x$ to the boundary of $\Omega$, in the direction $\omega$. 
Interchanging the order of integration between $r$ and $\rho$  in the right side of \eqref{eq:2} and performing the integral in $\rho$, we get
$$
\abs{f(x) - f_{\Omega}} 
\le \frac{1}{|\Omega|} \int_{|\omega|=1}    \di \omega \int_0^{\rho_{max}(\omega)}
\abs{\nabla f(x+r\omega)} 
\frac{1}{d}(\rho_{max}(\omega)^d - r^d)   \di r.
$$
Now, using that  $\rho_{max}(\omega)^d - r^d\le \beta^d$, where $\beta$ is the diameter of $\Omega$,  we get
\begin{eqnarray}
|f(x) - f_{\Omega}| \le \frac{\beta^d}{d |\Omega|} \int_{|\omega|=1} \di \omega\int_0^{\rho_{max}(\omega)}   
|\nabla f(x+r\omega) | \di r
=\frac{\beta^d}{d |\Omega|} \int_{\Omega}
r^{-(d-1)} |\nabla f(x+u) |  \di u.
\nonumber
\end{eqnarray}
Finally, using the definition of $h$, we get the desired inequality \eqref{eq:Davies}
\end{proof}

Now, we are ready to give the proof of Theorem~\ref{thm:GNSconvex}.

\begin{proof} [Proof of Theorem \ref{thm:GNSconvex}]
Let $d \ge 3$. Using H\"older's inequality we have
\begin{equation*}
\int_{\Omega} |u-u_{\Omega}|^{2+4/d}   \di x
\le \left(\int_{\Omega} |u-u_{\Omega}|^2 dx\right)^{2/d}
\left(\int_{\Omega} |u-u_{\Omega}|^{2d/(d-2)}   \di x\right)^{(d-2)/d}.
\label{eq:conv1}
\end{equation*}
Using Lemma \ref{average}
\begin{equation}
\int_{\Omega} |u-u_{\Omega}|^{2d/(d-2)}   \di x
\le
C_{\rm D}(\Omega, d)^{2d/(d-2)}
\left\||x|^{-(d-1)}*|\nabla u|\right\|_{2d/(d-2)}^
{2d/(d-2)}.
\label{eq:conv2}
\end{equation}
In $\mathbb{R}^d$, $|x|^{-d+1} \in L_{q,w} (\mathbb{R}^d)$ for $q=d/(d-1)$, and 
$$
\||x|^{-d+1} \|_{q,w} \equiv \sup_A \abs{A}^{-1/d} \int_A \abs{x}^{-d +1} \di x  = d \, \omega_d^{(d-1)/d},
$$
the volume of the unit ball in dimension $d$.
Then, using the Hardy-Littlewood-Sobolev inequality with 
$r=2d/(d-2)$, $p=2$ and $q=d/(d-1)$ following the notations of \cite[Theorem 4.3]{LiLo001}, 
we have 
\begin{align*}
 \||x|^{-(d-1)}*|\nabla u|\|_{2d/(d-2)} &
\le 
\dfrac{1}{d}  \, \omega_d^{-(d-1)/d} \, C_{\rm HLS}(d,d-1,2) \, \|\nabla u\|_2 
\, \||x|^{-(d-1)}\|_{q,w}  \\
&\le  C_{\rm HLS}(d,d-1,2) \, \|\nabla u\|_2 .
\end{align*}
Inserting this bound in \eqref{eq:conv2}, we obtain 
\[
\int_{\Omega} |u-u_{\Omega}|^{2d/(d-2)}  \di x
\le (C_{\rm D}(\Omega, d) \, C_{\rm HLS}(d,d-1,2) )^2 \left(\int_{\Omega} u^2  \di x \right)^{2/d} \int_{\Omega} |\nabla u|^2 \di x \qedhere
\]
\end{proof}

%
%

\section{A \textsc{GNS} inequality for cubes}
\label{sec:Rumin}
In this section, we use the explicit eigenfunctions of the Neumann Laplacian on the cube to obtain an improved inequality in this case.

\begin{theorem} \label{thm:rumin}
For all $d \ge 1$, we have the bound
\[
G_Q(d) \ge \frac{\pi^2 d^2}{(d+4)(d+2) N_d^{2/d}} \equiv G_2 (Q_d, d),
\]
where the constant $N_d$ is related to a counting problem (see \eqref{eq:def_Nd} below) and satisfies
\begin{equation} \label{eq:bound_Nd}
N_d \le \sum_{\ell=1}^{d} \binom{d}{\ell} \frac{\omega_{\ell}}{ \ell^{(d-\ell)/2}}.
\end{equation}
where $\omega_\ell$ is the volume of the $\ell$-dimensional unit ball.
\end{theorem}

The proof of this theorem is simpler than the proof for general convex domains. It follows closely the strategy in \cite[Theorem 4.26]{Lun017}. Theorem~\ref{thm:rumin} also gives better values for the constants. In Table~\ref{tbl:cubes}, we compare the numerical values for the constant in Theorem~\ref{thm:rumin} with the bound obtained for general convex domains. The first column contains the upper bound that will be proven in the next section.
\begin{table}[h!]
 \begin{tabular}{l | c | c | c}
 $d$ & $G(d)/ 4$ &  $G_1(Q_d, d)$ & $G_2(Q_d, d)$ \BSC \\ 
 \hline
 1 &   0.62     &   --  &   0.16\TSC  \\
 2 &  1.46    &     --    &  0.40           \\
3  &   2.39    &     0.1838    &    0.71      \\
4  &   3.37    &     0.0041    &   0.63     \\
5  &   4.37   &     0.0002    &    0.69\BSC   \\
\hline
\end{tabular}
\caption{\label{tbl:cubes} Comparison of the upper bound for $G_Q(d)$ (from the numerical values in Table~\ref{tbl:GNS}) with the lower bounds obtained in Theorem~\ref{thm:GNSconvex} and Theorem~\ref{thm:rumin}. See Remark~\ref{rmk:counting} for the values of $G_2(Q_d, d)$ for $d=2,3$.}
\end{table}

\begin{proof}
The starting point is the following representation of the gradient term, valid for all $u \in H^1(Q_d)$. We define $v = u - u_{Q_d}$ and write 
\begin{align*}
\int_{Q_d} \abs{\nabla u}^2= \int_{Q_d} \abs{\nabla v}^2 &= \sum_{\bvec k} E_{\bvec k} \abs{\inprodtwo{v}{ u_{\bvec k}}}^2 \\
& =\int_{0}^{\infty} \norm{P_{\ge E} v}^2 \di E, 
\end{align*}
where $u_{\bvec k}, E_{\bvec k}$ are the eigenfunctions and eigenvalues of the Neumann Laplacian on the cube, indexed by $\bvec k \in \N_0^d$, and $P_{\ge E}$ is the associated projector on energies below $E$.
Explicitly, 
\[
E_{\bvec k} = \pi^2 \abs{\bvec k}^2,  \quad u_{\bvec k} = C_{\bvec k}\cos(\pi k_1 x_1) \cos (\pi k_2 x_2) \cdots \cos(\pi k_d x_d),
\]
with normalization constant $C_{\bvec k}= 2^{\ell/2}$ with $\ell $ the number of nonzero components of $\bvec k$.
Now, we can bound
\begin{align*}
\abs {P_{ < E} v}^2(x) &= \Big( \sum_{\substack{\bvec k \in \N_0^d \\ 0 < \abs{\bvec k}^2 < E/\pi^2 }}   u_{\bvec k} (x )\inprodtwo{ u_{\bvec k}}{v}  \Big)^2 \\
& \le  \Big( \sum_{\substack{\bvec k \in \N_0^d \\ 0 < \abs{\bvec k}^2 < E/\pi^2 }}   \abs{u_{\bvec k} (x )}^2 \Big) \Big( \sum_{\bvec k \in \N_0^d} \abs{\inprodtwo{ u_{\bvec k}}{v}}^2 \Big) \\
& \le \norm{u}_2^2 \sum_{\substack{\bvec k \in \N_0^d \\ 0 < \abs{\bvec k}^2 < E/\pi^2 }}   C_{\bvec k}^2.
\end{align*} 
Note that there is no contribution of $\bvec k = 0$, since the average of $v$ vanishes by definition.
We will prove below that the inequality 
\begin{equation} \label{eq:def_Nd}
\sum_{\substack{\bvec k \in \N_0^d \\ 0 < \abs{\bvec k}^2 < E/\pi^2 }}   C_{\bvec k}^2 
\le N_d (E / \pi^2)^{d/2} 
\end{equation}
holds, for some finite constant $N_d$.
Assuming this for the moment, we write
\begin{align*}
 \int_{Q_d} \abs{\nabla v}^2 & = \int_{Q_d}\int_{0}^{\infty} \abs{P_{\ge E} v(x)}^2 \di E \di x \\
 & =  \int_{Q_d}\int_{0}^{\infty} \left(v(x) - P_{< E} v(x)\right)^2 \di E \di x \\
  & \ge \int_{Q_d}\int_{0}^{\infty} \left[\abs{v(x)} - \norm{v}E^{d/4} \pi^{-d/2} N_d^{1/2} \right]_+^2 \di E \di x \\
  & = \int_{Q_d} \abs{v(x)}^{2+4/d} \di x   \norm{v}^{-4/d} \pi^2 N_d^{-2/d} \int_0^1 (1-t^{d/4})^2 \di t,
\end{align*}
where the last equality follows by a change of variables.
The final integral gives a numerical constant
\begin{align*}
\int_0^1 (1-t^{d/4})^2 \di t &= \int_{-\infty}^0 (1-e^{ds/a})^2 e^s \di s \\
&= 1- \frac{2}{1+d/4}+\frac{1}{1+d/2} = \frac{d^2}{(d+2)(d+4)}.
\end{align*}
Since $\norm{v}_2 \le \norm{u}_2$, assuming \eqref{eq:def_Nd}, we have proven for all $u \in H^1(Q_d)$,
\[
\ \int_{Q_d} \abs{\nabla u}^2  \ge \int_{Q_d} \abs{u-u_{Q_d}}^{2+4/d} \left(\int_{Q_d} u^2 \right)^{-2/d} \frac{d^2}{(d+2)(d+4)} \pi^2 N_d^{-2/d}.
\]

It remains to prove that \eqref{eq:def_Nd} holds with $N_d$ given by \eqref{eq:bound_Nd}. 
Separating the summands according to the values of $C_{\bvec k}$, we find
\begin{align*}
\sum_{\substack{\bvec k \in \N_0^d \\ 0 < \abs{\bvec k} < r }}   C_{\bvec k}^2
&= \sum_{\ell=1}^d \binom{d}{\ell} 2^\ell \sum_{\substack{\bvec k \in \N^\ell \\ 0 < \abs{\bvec k} < r }} 1\\
& \le \sum_{\ell=1}^d \binom{d}{\ell} 2^\ell \frac{r^\ell \omega_\ell }{2^\ell} \chi (r \ge \sqrt \ell) \numberthis \label{eq:bound2_Nd}\\
& \le  r^d \sum_{\ell=1}^d \binom{d}{\ell} \frac{ \omega_\ell }{ \ell^{(d-\ell)/2}},
\end{align*}
where the first inequality bounds the number of integer points by the surface of the intersection of a ball with the first quadrant. 
\end{proof}

\begin{remark} \label{rmk:counting}
In lower dimensions it is possible to find better bounds on $N_d$, by using inequality \eqref{eq:bound2_Nd} for larger $r$ and checking \eqref{eq:def_Nd} explicitly for all lower values of $r$. This improves the bounds for $G_Q(d)$. An explicit verification for all $\bvec k \in \N_0^2$ with $\abs{\bvec k} \le 4 $ gives $N_2 \le \pi + 4/\sqrt{17} \approx 4.11 $.
For $d = 3$, an explicit verification for $\abs{\bvec k} \le \sqrt{18}  $ shows that $N_3 \le (4 \pi/3 + 3 \pi/\sqrt{19}+6/19 \approx 6.67 $.
With these values of $N_d$, we obtain the bounds
$G_Q(2) \ge 0.40$ and $G_Q(3) \ge 0.71$ respectively.
\end{remark}

 \section{Concentration and comparison in cubes}
\label{sec:GNS}

In this section, we prove the remaining statements of Theorem \ref{thm:meta}. We begin with the construction of test functions to obtain an upper bound on $G_Q(d)$.
\begin{lemma}
For all $d \ge 1$, we have $G_Q(d)  \le G(d)/4$.
\end{lemma}

\begin{proof}
Let $g= g(r)$ be a non-negative, spherically symmetric minimizer for the problem in $\R^d$ \eqref{GNSRDIntro}. Assume for the moment that $g \in L^1(\R^d)$. Then 
for $\lambda > 0$, we define
$$
u_\lambda (r  \omega) = \lambda^{d/2} g(\lambda r), 
$$
where $ \omega \in \mathbb{S}^{d-1}$. 
By scaling and by using that $2^d$ copies of $Q_d= [0,1]^d$ cover $[-1,1]^d$, we find
\begin{align*}
0 \le \int_{Q_d} u_\lambda    \equiv \overline{u_\lambda} &\le \lambda^{-d/2} 2^{-d}\int_{\R^d} g  ,\\
\int_{Q_d} u_\lambda^2  & \le  2^{-d}\int_{\R^d} g^2  ,\\
\int_{Q_d} \abs{\nabla u_\lambda}^2   &\le \lambda^2 2^{-d}\int_{\R^d} \abs{\nabla g}^2  .
\end{align*}
For the denominator, we bound
\begin{align*}
\int_Q \abs{u_\lambda-\overline{u_\lambda}} ^{2+4/d}  
&\ge \int_Q u_\lambda^{2+4/d}    -(2+4/d)\overline{u_\lambda} \int_Q u_\lambda^{1+4/d}   \\
&\ge \lambda^2 2^{-d}\int_{B(0, \lambda)} g^{2+4/d}   - C\lambda^{-d + 2} \int_{\R^d} g   \int_{\R^d} g^{1+4/d}  \\
& \ge \lambda^2 2^{-d}\int_{B(0, \lambda)} g^{2+4/d}   -C\lambda^{-d + 2} \left(\int_{\R^d} g   \right)^{\frac{2d+4}{d+4}} \left(\int_{\R^d} g^{2+4/d}  \right)^{\frac{4}{d+4}},
\end{align*}
where $C$ is a positive constant.
Thus, we obtain
\begin{align*}
G_Q(d) 
&\le \liminf_{\lambda \to \infty } \frac{\int_Q \abs{\nabla u_\lambda}^2    \left(\int_Q u_\lambda^2    \right)^{2/d}}{\int_Q \abs{u_\lambda-\overline{u_\lambda}} ^{2+4/d}   } \\
& \le \left(2^{-d} \right)^{2/d} \frac{\int_{\R^d} \abs{\nabla g}^2    \left(\int_{\R^d} g^2    \right)^{2/d}}{\int_{\R^d} g^{2+4/d}  } = G(d)/4.
\end{align*}
If $g$ is not in $L^1(\R^d)$, then we may apply the same strategy to $[g-\epsilon]_+$, and take $\epsilon \to 0$ after having taken $\lambda \to \infty$.
\end{proof}

In the one-dimensional case, we obtain the corresponding lower bound as well.
\begin{theorem} In one dimension, we have
$$
G_Q(1) =  \frac{G(1)}{4} = \frac{\pi^2}{16},
$$
and the infimum is not attained.
\end{theorem}

\begin{proof} We will write $Q_1 = I$.
Fix $u\in H^1(I) $.
Upon replacing $u$ by $u-u_I$, we may assume $ u_I=0$. 
Denote by $u_+\equiv \max(u,0)$ and $u_-\equiv \max(-u,0)$ the positive and negative parts of $u$, and by $u_\pm^*$ their nonincreasing rearrangements.
We construct $f_\pm \in H^1(\R)$ by reflecting $u_\pm^*$ with respect to $0$.
By applying the GNS inequality (\ref{GNSRDIntro}), with $d=1$,  to $f_\pm$, we find
\begin{align*}
 G(1) 
 &\le \frac{\int_\R(f'_\pm)^2  \left( \int_\R f_\pm^2\right)^2}{\int_\R f_\pm^6}  \\
 & = \frac{2 \int_I ({u^*_\pm}')^2 \left( 2\int_I {u^*_\pm}^2\right)^2}{2\int_I {u^*_\pm}^6 } \\
 &\le 4\frac{ \int_I (u'_\pm)^2 \left( \int_I u_\pm^2\right)^2}{\int_I u_\pm^6 }.
\end{align*}
Now we can bound
\begin{align*}
\frac{ \int_I (u')^2 \left( \int_I u^2\right)^2}{\int_I u^6 } 
&= \frac{ \left( \int_I (u'_+)^2 +\int_I(u'_-)^2 \right)\left( \int_I u_+^2+ \int_I u_-^2\right)^2}{\int_I u_+^6+ \int_Iu_-^6 } \\
&\ge \frac{ \int_I (u'_+)^2\left( \int_I u_+^2\right)^2 +\int_I(u'_-)^2 \left( \int_I u_-^2\right)^2}{\int_I u_+^6+ \int_Iu_-^6 } \numberthis \label{eq:cross-terms}\\
&\ge \frac{G(1)}{4}
\end{align*}
If the infimum would be attained, for some function $u$ with $\int_I u = 0$ there would be equality in \eqref{eq:cross-terms}. 
This implies that
\begin{align*}
\int_I (u_+')^2 \int_I u_-^2 = 0 \quad \text{and } \int_I u_+ = -\int_I u_-,
\end{align*}
which is impossible for $u \neq 0$.
\end{proof}

Finally, for dimensions $d \ge 1$, we obtain the following dichotomy.
\begin{theorem} \label{thm:dichotomy}
 For $d \ge 2$, if a minimizer for $G_Q(d)$ does not exist, then 
 \[
G_Q (d) = G(d)/4.
 \]
\end{theorem}
\begin{remarks}
 \begin{enumerate}[i)]
  \item Numerical simulations suggest that in $d=2$, minimizers do not exist due to concentration. This suggests the conjecture that $G_Q(2) = G(2)/4$.
  \item The reason why we are not able to obtain the sharp constants $G_Q(d)$ for $d>1$ is precisely because the rearrangement lemma \ref{lem:rearrangements} does only hold for functions supported in small sets, while in $d=1$ a rearrangement inequality is available for all nonnegative functions.
 \end{enumerate}

\end{remarks}

The proof of this theorem relies on the following rearrangement lemma.

\begin{lemma}\label{lem:rearrangements}
For all $d \ge 1$, there exists $V_d >0$ such that, 
 if $u \in H^1(Q_d)$ is nonnegative and $\abs{\supp u} \le V_d$, then 
 \[
 \int_{Q_d} \abs{\nabla u^*}^2 \le  \int_{Q_d} \abs{\nabla u}^2, 
 \]
 where $u^*$ is the rearrangement of $u$ such that its level sets are the intersections of $Q_d= [0,1]^d$ with a (hyper)-sphere centered at the origin.
\end{lemma}
This lemma will be proven in the appendix.

\begin{proof}[Proof of Theorem~\ref{thm:dichotomy}]
 Consider a minimizing sequence $(v_n)$ for \eqref{eq:var-prob}, normalized such that $\int_{Q_d} v_n = 0$ and $\int_{Q_d} v_n^2 = 1$. If $\norm{v_n}_{2+4/d}$ is bounded uniformly in $n$, then the sequence is bounded in $H^1({Q_d})$ and therefore has a subsequence that converges weekly in to some $v$. By the Rellich-Kondrachov theorem, this subsequence converges strongly in $L^1({Q_d})$, in $L^2({Q_d})$ and $L^{2+4/d}({Q_d})$ (note that $2+4/d < 2d/(d-2)$). On the other hand, using Fatou's lemma,
 \begin{align*}
  G_Q(d)  &\le \frac{\int_{Q_d} \abs{\nabla v}^2    \di x  \left(\int_{Q_d} v^2    \di x\right)^{2/d}}{ \int_{Q_d} \abs{v}^{2+4/d}  \di x } \\
  &\le \lim_{n \to \infty}\frac{\int_{Q_d} \abs{\nabla v_n}^2   \di x \left(\int_{Q_d} v_n^2    \di x \right)^{2/d}}{ \int_{Q_d} \abs{v_n}^{2+4/d}  \di x} = G_Q(d),
 \end{align*}
and $v$ is a minimizer.

 Thus, we may assume that $\norm{v_n}_{2+4/d} \to + \infty$ and show that in this case $G_Q(d) = G(d)/4$.
 We define $m_n= \norm{v_n}_{2+4/d}^{1/2} $ and consider as test functions
 $$
 u_n = [v_n-m_n]_+ -[v_n+m_n]_-.
 $$
 Note that
 \[
 \operatorname{ess\,sup}{\abs{v_n}} \ge  \norm{v_n}_{2+4/d} = m_n^2,
 \]
 so $u_n \neq 0$ as soon as $m_n > 1$.
We first show that $u_n$ is a minimizing sequence as well. We have that
\begin{align*}
 \int_{Q_d} u_n^2   &\le \int_{Q_d} v_n^2   ,  \qquad \int_{Q_d} \abs{\nabla u_n}^2    \le \int_{Q_d} \abs{\nabla v_n}^2   , \\
 \int_{Q_d} u_n   &= \int_{\{v_n \ge m_n\}} (v_n - m_n)    + \int_{\{v_n \le -m_n\}} (v_n + m_n)   \\
 & = \int_Q v_n    - \int_{\{\abs{v_n} < m_n\} } v_n   - m_n \abs{\{v_n \ge m_n\}} + m_n  \abs{\{v_n \le - m_n\}}   ,
\end{align*}
so since $\int_Q v_n   = 0$,
\begin{equation*}
\abs{\overline{u_n}} \equiv \abs{\int_Q u_n   }  \le m_n .
\end{equation*}
This means that
\begin{equation} \label{eq:bound-convexity}
 \norm{u_n- \overline u_n}_{2+4/d} \ge \norm{u_n}_{2+4/d} - m_n \ge \norm{v_n}_{2+4/d} -2 m_n = \norm{v_n}_{2+4/d} (1-2m_n^{-1}) ,
\end{equation}
so we can bound
$$
\int_Q \abs{u_n- \overline{u_n}}^{2+4/d}      \ge \int_Q \abs{v_n}^{2+4/d}     \left(1 - C m_n^{-1}  \right). 
$$
Combining these bounds, we find that 
$$
\lim_{n \to \infty} \frac{\int_{Q} \abs{\nabla u_n}^2     \left(\int_{Q} u_n^2     \right)^{2/d} }{\int_Q \abs{u_n- \overline{u_n}}^{2+4/d}    }
\le  \lim_{n \to \infty}\frac{\int_{Q} \abs{\nabla v_n}^2      \left(\int_{Q}v_n^2     \right)^{2/d}}{ \int_Q\abs{v_n}^{2+4/d}    \left(1 - C m_n^{-1}  \right)} = G_Q(d),
$$
so the sequence $u_n$ is a minimizing sequence as well. On the other hand, the support of $u_n$ becomes small, since
$$
1 = \int_Q v_n^2      \ge m_n \abs{\supp u_n}. 
$$
For $n$ large enough, $\abs{\supp u_n} \le V_d$ and we may apply Lemma~\ref{lem:rearrangements} separately to the positive and negative parts of $u_n$. The resulting rearranged 
function can be extended to $\R^d$ as a spherically symmetric function and therefore obeys the corresponding Gagliardo--Nirenberg inequality. As before, the spherically symmetric 
function contains $2^d$ copies of the original one, hence the factor $4 = 2^d \times (2^{d})^{2/d} / 2^d$.
We have
\begin{align*}
 &\frac{\int_{Q} \abs{\nabla u_n}^2     \left(\int_{Q} u_n^2      \right)^{2/d} }{\int_Q \abs{u_n- \overline u_n}^{2+4/d}   }  \\
 & \quad \ge \frac{\int_{Q} \abs{\nabla u_{n,+}}^2     \left(\int_{Q} (u_{n,+})^2     \right)^{2/d} + \int_{Q} \abs{\nabla u_{n,-}}^2     
 \left(\int_{Q} (u_{n,-})^2 \right)^{2/d}     }{\int_Q \abs{u_n}^{2+4/d} 
     \left(1-C m_n^{-1} \right)}\\
 & \quad \ge \frac{\int_{Q} \abs{\nabla u_{n,+}^*}^2      \left(\int_{Q} (u_{n,+}^*)^2     \right)^{2/d} + \int_{Q} \abs{\nabla u_{n,-}^*}^2    
  \left(\int_{Q} (u_{n,-}^*)^2     \right)^{2/d}  }{\left(\int_Q \abs{u_{n,+}^*}^{2+4/d}     +\int_Q \abs{u_{n,-}^*}^{2+4/d}     \right)\left(1-m_n^{-1} \right)}\\
 &\quad \ge \frac{G(d)}{4} \left(1-m_n^{-1} \right)^{-1}.
\end{align*}
which shows, by taking the limit $n \to \infty$,
$$
G_Q(d)  \ge \frac{G(d)}{4}.
$$
The opposite inequality has been proven above.
\end{proof}

\appendix
\section{Proof of Lemma~\ref{lem:rearrangements} (Rearrangements)}

The proof of Lemma~\ref{lem:rearrangements} relies on the following result about sets minimizing the perimeter for a given volume in $Q_d$. In this context, the \emph{volume} of a subset of $Q_d$ is its $d$-dimensional Lebesgue measure and the \emph{perimeter} is the $(d-1)$-dimensional Hausdorff measure of its boundary in $Q_d$. 

\begin{theorem}[Perimeter minimizers for small volume are balls, \cite{MoJo000}]
For each $d \ge 1$, there exist $V_d>0$ such that for $V \le V_d$, sets of volume $V$ minimizing the perimeter in $Q_d$ are intersections of $Q_d$ with a ball centered in a corner. 
\end{theorem}

For $d= 2$, a computation comparing discs with rectangles shows that $V_2 = 1/\pi$. For $d = 3$, it is conjectured that minimizers of perimeter are balls centered at corners, cylinders centered at an edge and cuts of the cube by a halfplane. If this conjecture is true, $V_3= \pi/ 3^4$. For larger $d$, the problem is even more difficult and the proof of the theorem relies on a compactness argument.

\begin{proof}
In the flat, $d$-dimensional torus ($[-1,1]^d$ with periodic boundary conditions), a standard reflection argument, see for instance Sec. 1.3 of the review \cite{Ro005}, shows that the minimizers of perimeter are symmetric under a reflection of each axis. 
 Therefore, the problem of minimizing the perimeter enclosing a volume $V$ in $[0,1]^d$ is equivalent to minimizing the perimeter of a set of volume $2^d V$ in the torus. This is a smooth, compact manifold without curvature so the result follows from \cite[Theorem 4.4, case b]{MoJo000}.
\end{proof}

Assuming the isoperimetric result, we can follow the strategy of Talenti \cite[sec.\,1.5]{Tal994} to obtain the rearrangement inequality for gradients.

\begin{proof}[Proof of Lemma~\ref{lem:rearrangements}] This proof follows exactly the same steps as the proof of the rearrangement inequality in $\R^d$ from \cite[sec.\,1.5]{Tal994}. It is reproduced here for the sake of completeness.
For shortness, we write $\kappa_d  = \omega_d / 2^d$, the volume of a unit ball centered at the origin intersected with $Q_d$. The Lebesgue measure of a set $E \subset Q_d$ will be denoted by $\mu(E)$ and $\mathcal{H}^{d-1}(E)$ will be the $d-1$-dimensional Hausdorff measure. With these definitions, the isoperimetric inequality for $E \subset Q_d$ with $\mu(E) \le V_d \le \kappa_d$ becomes
\[
\mathcal{H}^{n-1}(\partial E) \ge d \kappa_d^{1/d} \left( \mu(E)\right)^{1-1/d}.
\]
Now fix $u \in H^1(Q_d)$ nonnegative such that $\abs{\supp u} \le V_d$. We write the rearrangement $u^*$ in the form
\[
u^*(x) = v(\kappa_d \abs{x}^d).
\]
The function $v$ is non-increasing, continuous and maps $[0, \abs{\supp ( u)}  ]$ onto $[0, \operatorname{ess\,sup} (u)]$. 
As a first step we use spherical coordinates and a change of variables (relating radius $r$ to volume fractions $s$) to write
\begin{align}
\int_{Q_d} \abs{\nabla u^*}^2 = \int_{0}^1 d \kappa_d r^{d-1} \left(\frac{d}{d r} u^*(r) \right)^2\di r \nn \\
= \int_{0}^1 \di s \left( v'(s) d \kappa_d^{1/d} s^{1-1/d}\right)^2 \label{eq:calc-spherical}. 
\end{align}  
On the other hand, we use the fundamental theorem of calculus to express the right hand side of the inequality as a function of $s$,
\begin{align*}
 \int_{Q_d} \abs{\nabla u}^2 = \int_0^1 \di s \frac{d}{d s} \int_{\{x | v(s) < u(x)\}} \abs{\nabla u(x)}^2 \di x \\
=  \int_0^1 \di s \lim_{h \to 0} \frac{1}{h} \int_{\{x | v(s+h) < u(x) < v(s)\}} \abs{\nabla u(x)}^2 \di x.  
\end{align*}
From this point on, we may concetrate on those values of $s$ such that $v(s)$ is strictly decreasing in a neighborhood of $s$. The contribution to both integrals of the values of $s$ such that this is not the case vanishes. This also means that for small $h$ we don't have to distinguish $\mu \{ x |u(x) > v(s+h)\}$ and $\mu \{x | u(x) \ge v(s+h)\}$ etc.
By the Cauchy-Schwarz inequality,
\begin{align*}
\left(\int_{\{v(s+h) <u(x) < v(s)\}} \abs{\nabla u(x)} \di x \right)^2&\le \mu \{x| v(s+h) < u(x) <v(s)\}   \int_{\{v(s+h) \le u(x) \le v(s)\}} \abs{\nabla u}^2 \di x \\
&\le h  \int_{\{v(s+h) \le u(x) \le v(s)\}} \abs{\nabla u}^2 \di x , 
\end{align*}
where the last line follows from the definition of $v$. We obtain
\begin{align} \label{eq:iso-a}
 \frac{d}{d s} \int_{\{x | v(s) < u(x)\}} \abs{\nabla u(x)}^2 \di x \ge 
 \left( \lim_{h \to 0 } \frac{1}{h}\int_{\{v(s+h) \le u(x) \le v(s)\}} \abs{\nabla u} \di x \right)^2. 
\end{align}
In order to estimate the integrand, we use the coarea formula to write
\begin{align*}
\int_{\{v(s+h) \le u(x) \le v(s)\}} \abs{\nabla u} \di x& = \int_{v(s+h)}^{v(s)}  \mathcal{H}^{n-1}(\{x|u(x) = t\}) \di t \\
& \ge \int_{v(s+h)}^{v(s)}  d \, \kappa_d^{1/d } (\mu\{x| u(x) > t\})^{1-1/d} \di t.
\end{align*}
Here, the last line is where the isoperimetric inequality comes into play, to compare the area ($n-1$-dimensional Hausdorff measure) of the boundary of the set $\{x | u(x) > t\}$ to its volume.
By definition of the rearrangement, for $t \le v(s)$,
\[
\mu\{x| u(x) > t\} \ge \mu\{x| u(x) > v(s)\} = s.
\]
Therefore, we obtain the estimate
\begin{align*}
\lim_{h \to 0 } \frac{1}{h}\int_{\{v(s+h) \le u(x) \le v(s)\}} \abs{\nabla u} \di x   
&\ge d \, \kappa_d^{1/d } s^{1-1/d} \lim_{h \to 0} \frac{v(s)-v(s+h)}{h} \\
&=d \, \kappa_d^{1/d } s^{1-1/d} (-v'(s)).
\end{align*}
Recall that $v$ is decreasing. Inserting this in the inequality \eqref{eq:iso-a} and comparing with
\eqref{eq:calc-spherical} gives the result.
\end{proof}

\bigskip
\bigskip

\section*{Acknowledgments}
\thanks{The work of R.B. has been supported by Fondecyt (Chile) Projects \# 116--0856 and \#114-1155.  The work of C.V. has been supported by a ``Beca Presidente de la 
Rep\'ublica" (Chile) fellowship and by a ``Beca Padre Hurtado" (PUC) fellowship. The work of H. VDB. has been partially supported by  CONICYT (Chile)  (PCI) project REDI170157 and partially by
Millennium Nucleus ``Center for Analysis of PDE'' NC130017.


\begin{thebibliography}{10}

\bibitem{AcDu003}
G.~Acosta and R.~G.~Duran, 
{\it An Optimal Poincar\'e Inequality in $L^1$ for Convex Domains}, 
Proc. A.~M.~S. {\bf 132} (2003) 195--202.

\bibitem{AdSeTi017}
R.~Adami, E.~Serra, and P.~Tilli, 
{\it Negative Energy Ground States for the $L^2$--critical NLSE on Metric Graphs}, 
Commun.~Math.~Phys. {\bf 352} (2017) 387--406.


\bibitem{Agu008}
M.~Agueh, {\it Gagliardo--Nirenberg inequalities involving the gradient $L^2$--
norm}, Comptes Rendus Mathematique {\bf 346} (2008) 757--762. 
 
\bibitem{Bab961}
K.~I.~Babenko, 
{\it An inequality in the theory of Fourier
integrals}, 
Izv.~Akad.~Nauk~SSSR Ser.~Mat., {\bf 25} (1961) 531--542.
The English translation is in Amer.~Math.~Soc.~Transl. {\bf 44}  (1961) 115--128.
 
\bibitem{Bec975}
W.~Beckner, 
{\it Inequalities in Fourier analysis}, Annals of Mathematics {\bf 102} (1995) 159--182.
 
 
\bibitem{BeLo004}
R.~D.~Benguria, and M.~Loss, 
{\it Connection between the Lieb--Thirring conjecture for Schr\"odinger operators and an isoperimetric problem for ovals on the plane}, 
Contemporary Mathematics {\bf 362} (2004) 53--61. 
 
\bibitem{BeIlNi975}
O.~V.~Besov, V.~P.~Il'in, and S.~M.~Nikol'skii, {\bf Integral representations of functions and imbedding
theorems}, Nauka, Moscow, 1975. (The English transiation was published in two volumes by J.~Wiley~Sons, NY, 1979).
 
 
 
\bibitem{Dav990}
E.~Brian~Davies, 
{\it Heat Kernels and Spectral Theory},
Cambridge Tracts in Mathematics {\bf 92}
Cambridge University Press, Cambridge, 1990. 

\bibitem{DoFeLoPa006}
J.~Dolbeault, P.~Felmer, M.~Loss, and E.~Paturel, 
{\it Lieb--Thirring type inequalities and
Gagliardo?Nirenberg inequalities for systems}, 
Journal of Functional Analysis {\bf 238} (2006) 193--220.

\bibitem{Fra014}
R.~L.~Frank, 
{\it Cwikel's theorem and the CLR inequality}, 
Journal of Spectral Theory {\bf 4} (2014) 1-21. 

\bibitem{Gag959}
E.~Gagliardo,
{\it Ulteriori propiet\`adi alcune classi di funzioni on pi\`u variabli}, 
Ric.~Mat., {\bf 8} (1959) 24--51.


\bibitem{Lie983}
E.~H.~Lieb, 
{\it Sharp Constants in the Hardy--Littlewood--Sobolev and Related Inequalities}, 
Annals of Mathematics {\bf 118}  (1983) 349--374.

\bibitem{LiLo001}
E.~H.~Lieb and M.~Loss,
{\it  Analysis, Second Edition},
Graduate Studies in Mathematics,  {\bf 14} 
American Mathematical Society, Providence, 2001.

\bibitem{LiTh976}
E.~H.~Lieb and W.~Thirring, 
{\it  Inequalities for the moments of
the eigenvalues of the
Schr\"odinger Hamiltonian and their relation to Sobolev inequalities}, 
in {\bf Studies in Mathematical Physics}, Essays in Honor of Valentin Bargmann, 
E.~H.~Lieb, B. Simon, A.~S.~Wightman Eds., 
pp. 269--303, Princeton University Press, 1976.


\bibitem{Lun017}
D.~Lundholm, 
{\it Methods of Modern Mathematical Physics: Uncertainty and Exclusion Principles in Quantum Mechanics}, 
Lecture Notes for a Master Class Course at KTH, Stockholm, Spring 2017. 


\bibitem{LuSo013}
D.~Lundholm and J.~P.~Solovej, 
{\it Hardy and Lieb--Thirring inequalities for anyons}, 
Commun.~Math.~Phys. {\bf 322} (2013) 883--908.


 \bibitem{MiTaSeOi017}
M.~Mizuguchi, K.~Tanaka, K.~Sekine, and S.~Oishi, 
{\it Estimation of Sobolev embedding constant on a domain dividable into bounded
convex domains},  J.~Inequal.~Appl.
2017:299 (2017). 
 
\bibitem{MoJo000}
F.~Morgan and D.~L.~Johnson, 
{\it Some sharp isoperimetric theorems for
Riemannian manifolds}, 
 Indiana University Mathematics Journal {\bf 49} (2000) 1017--1041.
 
 \bibitem{Nam018}
 Phan--Th\`anh~Nam, 
 {\it Lieb--Thirring inequalities with semiclassical constant and gradient correction}, 
 Journal of Functional Analysis (in press) (2018).
 
\bibitem{Nas989}
S.~M.~Nasibov,
{\it Optimal constants in some Sobolev inequalities and their ap
plications to the nonlinear Schr\"odinger equation},
Dokl.~Akad.~Nauk SSSR, {\bf 307} (1989) 538--542 (in Russian).
The English translation is in Soviet Math. Dokl. 40 (1990), 110--115
 
 \bibitem{Nir959}
 L.~Nirenberg,
{\it On elliptic partial differential equations: Lecture II}, 
Ann.~Sc.~Norm.~Super.~Pisa {\bf 13} (1959) 115--162.

 

\bibitem{PaWe960}
L.~E.~Payne, H.~F.~Weinberger, {\it An optimal Poincar\'e inequality for convex domains}, Arch.
Rat. Mech. Anal. {\bf 5} (1960) 286-292.
 
\bibitem{Ro005}
A.~Ros, {\it The isoperimetric problem} In {\bf Global theory of minimal surfaces}, Clay Math. Proc. {\bf 2} (2005) 175-209.

\bibitem{Ru011}
M.~Rumin, 
{\it Balanced distribution--energy inequalities and related entropy
bounds}, 
Duke Math. J.  {\bf 160} (2011) 567--597.
 
 \bibitem{Sol011}
 J.~P.~Solovej, 
 {\it Examples: Non-interacting systems and the Lieb-Thirring inequality, Charged systems}, 
 Lecture 2 of the course {\it Spectral Theory of N-body Schrödinger operators} given at the 
 University College London, Spring 2011. 
 http:/\!/www.ucl.ac.uk/$\sim$ucahipe/solovej-lt.pdf
 
\bibitem{Tal994}  
G.~Talenti, {\it Inequalities in rearrangement invariant function spaces},  In
{\bf Nonlinear Analysis, Function Spaces and Applications},
Miroslav Krbec, Alois Kufner, Bohumír Opic, and Jirí Rákosník, editors,
Proceedings of the Spring School held in Prague, May 23--28, 1994, volume 5,  
Prometheus Publishing House, Prague, 1994.



\bibitem{Zug017}
S.~Zugmeyer, 
{\it Sharp trace Gagliardo-Nirenberg-Sobolev inequalities for convex cones, and convex domains}, 
preprint (2017), ArXiv:1710.08233.


\end{thebibliography}

\end{document}